\documentclass{article}
\usepackage{ijcai17}
\usepackage{ctable}
\usepackage{subfig}
\usepackage{amsmath,amssymb,amsthm}
\usepackage{float}
\usepackage[vlined,ruled,linesnumbered]{algorithm2e}

\usepackage{times}
\usepackage{xspace}

 \newcommand{\F}{\mathcal{F}}
 \renewcommand{\H}{\mathcal{H}}

 \newcommand{\X}{\mathcal{X}}
\newcommand{\Y}{\mathcal{Y}} 

\newcommand{\LTL}{{\sc ltl}\xspace}
\newcommand{\LTLf}{{\sc ltl}$_f$\xspace}

\newcommand{\LDLf}{{\sc ldl}$_f$\xspace}

\newcommand{\Nat}{{\rm I\kern-.23em N}}

\def\ltlf{\LTLf}
\def\tt{\top}
\def\ff{\bot}
\def\X{\mathcal{X}}
\def\Y{\mathcal{Y}}
\def\F{\mathcal{G}}

\newtheorem{definition}{Definition}
\newtheorem{theorem}{Theorem}

\title{Symbolic \LTLf Synthesis}

\author{Shufang Zhu\\East China Normal University \And Lucas M. Tabajara\\Rice University \And Jianwen Li\thanks{Corresponding author}\\East China Normal University\\ \& Rice University 
\AND Geguang Pu\thanks{Corresponding author}\\East China Normal University \And Moshe Y. Vardi\\Rice University}

\begin{document}
\maketitle

\begin{abstract}
\LTLf synthesis is the process of finding a strategy that satisfies a linear temporal specification over finite traces. An existing solution to this problem relies on a reduction to a DFA game. In this paper, we propose a symbolic framework for \LTLf synthesis based on this technique, by performing the computation over a representation of the DFA as a boolean formula rather than as an explicit graph. This approach enables strategy generation by utilizing the mechanism of boolean synthesis. We implement this symbolic synthesis method in a tool called \emph{Syft}, and demonstrate by experiments on scalable benchmarks that the symbolic approach scales better than the explicit one.
\end{abstract}

\section{Introduction}\label{sec:intro}



The problem of synthesis from temporal specifications can be used to model a number of different problems in AI, in particular planning. Linear Temporal Logic, or \LTL, is the standard formalism to describe these specifications, but while \LTL is classically interpreted over infinite runs, many planning problems have a \emph{finite horizon}, since they assume that execution stops after a specific goal is achieved. This context leads to the emergence of a version of \LTL with alternative semantics over finite traces, called \LTLf. Some of the planning problems which can be reduced to \LTLf synthesis include several variants of conditional planning with full observability~\cite{BacchusK00,GHL09}. A general survey of applications of \LTLf in AI and CS can be found in~\cite{DegVa13}. 

\LTL synthesis in infinite-horizon settings has been well investigated in theory since~\cite{PnueliR89}, but the lack of good algorithms for the crucial step of automata determinization is prohibitive for finding practical implementations~\cite{DFogartyKVW13}. In addition, usual approaches rely on parity games~\cite{Thomas95}, for which no polynomial-time algorithm is known. In contrast, in the finite-horizon setting the specification can be represented by a finite-state automaton, for which determinization is practically feasible~\cite{TRV12} and the corresponding game can be solved in polynomial time on the number of states and transitions of the deterministic automaton. This opens the possibility for theoretical solutions to be implemented effectively.


A solution to the \LTLf synthesis problem was first proposed in~\cite{DegVa15}, based on DFA games. Following this method, an \LTLf specification can be transformed into a DFA with alphabet comprised of propositional interpretations of the variables in the formula. A winning strategy for the game defined by this DFA is then guaranteed to realize the temporal specification.

In this paper, we present the first practical implementation of this theoretical framework for \LTLf synthesis, in the form of a tool called \emph{Syft}. \emph{Syft} follows a symbolic approach based on an encoding of the DFA using boolean formulas, represented as Binary Decision Diagrams (BDDs), rather than an explicit representation through the state graph. We base the choice of a symbolic approach in experiments on DFA construction from an \LTLf specification. We compared between two methods for DFA construction: one symbolic using the tool MONA~\cite{KlaEtAl:Mona}, receiving as input a translation of the \LTLf specification to first-order logic, and one explicit using SPOT~\cite{spot}. Although both methods display limited scalability, the results show that the symbolic construction scales significantly better.

Using a symbolic approach allows us to leverage techniques for \emph{boolean synthesis}~\cite{DrLu.cav16} in order to compute the winning strategy. Our synthesis framework employs a fixpoint computation to construct a formula that expresses the choices of outputs in each state that move the game towards an accepting state. By giving this formula as input to a boolean synthesis procedure we can obtain a winning strategy whenever one exists.

Further experiments comparing the performance of the \emph{Syft} tool with an explicit implementation \emph{E-Syft} again display better scalability for the symbolic approach. Furthermore, a comparison with a standard \LTL synthesis tool confirms that restricting the problem to finite traces allows for more efficient techniques to be used. Finally, the results show DFA construction to be the limiting factor in the synthesis process.

\section{Preliminaries}
\subsection{\ltlf Basics}
Linear Temporal Logic (\LTL) over finite traces, i.e. \ltlf, has the same syntax as \LTL over infinite traces introduced 
in \cite{Pnu77}. Given a set of propositions $P$, the syntax of \ltlf formulas is defined as follows:
\begin{center}
	$\phi ::= \tt\ |\ \ff\ |\ a\ |\ \neg \phi\ |\ \phi_1\wedge\phi_2\ |\ X\phi\ |\ \phi_1 U \phi_2$
\end{center}

$\tt$ and $\ff$ represent \textit{true} and \textit{false} respectively. $a\in P$ is an 
\textit{atom}, and we define a literal $l$ to be an atom or the negation of an atom. $X$ (Next) and $U$ (Until)
are temporal operators. We also introduce their dual operators, namely $X_w$ (Weak Next) and $R$ (Release), defined as 
$X_w \phi\equiv \neg X\neg \phi$ and $\phi_1 R\phi_2\equiv \neg (\neg\phi_1 U\neg \phi_2)$. 
Additionally, we define the abbreviations $F\phi\equiv \tt U\phi$ and $G\phi\equiv \ff R\phi$. Standard boolean abbreviations, such as $\vee$ (or) and $\rightarrow$ (implies) are also used. 

A \textit{trace} $\rho = \rho[0],\rho[1],\ldots$ is a sequence of propositional interpretations (sets), in which 
$\rho[m]\in 2^P$ ($m \geq 0$) is the $m$-th interpretation of $\rho$, and $|\rho|$ represents the length of $\rho$. Intuitively, $\rho[m]$ is interpreted as the set of propositions which are $true$ at instant $m$.
Trace $\rho$ is an \textit{infinite} trace if $|\rho| = \infty$, which is formally denoted as $\rho\in (2^P)^{\omega}$; 
 otherwise $\rho$ is a \textit{finite} trace, denoted as $\rho\in (2^P)^{*}$. \ltlf formulas are interpreted over finite traces. Given a finite trace $\rho$ and an \ltlf formula 
$\phi$, we inductively define when $\phi$ is $true$ for $\rho$ at step $i$ ($0 \leq i < |\rho|$), written $\rho, i \models \phi$, as follows: 
\begin{itemize}
  \item $\rho, i \models\tt$;
  \item $\rho, i \not\models\ff$;
  \item $\rho, i \models a$ iff $a \in \rho[i]$;
  \item $\rho, i \models \neg \phi$ iff $\rho,i \not\models \phi$;
  \item $\rho, i \models\phi_1 \wedge \phi_2$, iff $\rho,i \models \phi_1$ and $\rho, i \models \phi_2$;
  \item $\rho, i \models X\phi$, iff $i+1 < |\rho|$ and $\rho, i+1 \models \phi$;
  \item $\rho, i \models \phi_1 U \phi_2$, iff there exists $j$ s.t. $i\leq j < |\rho|$ and $\rho, j\models \phi_2$, and for all $k$, $i \leq k < j$, we have $\rho, k \models \phi_1$.
\end{itemize}

An \ltlf formula $\phi$ is $true$ in $\rho$, denoted by $\rho \models \phi$, if and only if $\rho, 0\models\phi$. We next define the \ltlf synthesis problem.

\begin{definition}[\ltlf Synthesis]
Let $\phi$ be an \ltlf formula and $\X, \Y$ be two disjoint atom sets such that $\X\cup \Y = P$. $\X$ is the set of \emph{input variables} and $\Y$ is the set of \emph{output variables}. $\phi$ is \emph{realizable} with respect to $\langle\X, \Y\rangle$ if there exists a strategy $g: (2^{\X})^*\rightarrow 2^{\Y}$, such that for an arbitrary infinite sequence $\lambda = X_0,X_1,\ldots\in (2^{\X})^{\omega}$ of propositional interpretations over $\X$, we can find $k\geq 0$ such that $\phi$ is $true$ in the finite trace $\rho = (X_0\cup g(\epsilon)), (X_1\cup g(X_0)),\ldots , (X_k\cup g(X_0,X_1,\ldots,X_{k-1}))$.
\end{definition}
Moreover, variables in $\X$ ($\Y$) are  called X- (Y-)variables.

\subsection{DFA Games}
The traditional way of performing \ltlf synthesis is by a reduction to corresponding Deterministic Finite Automaton (DFA) games. According to~\cite{DegVa15}, every \ltlf formula can be translated in 2EXPTIME to a DFA which accepts the same language as the formula. 

Let the DFA $\F$ be $(2^{\X\cup \Y}, S, s_0, \delta, F)$, where $2^{\X\cup \Y}$ is the alphabet, $S$ is the set of states, $s_0$ is the initial state, $\delta: S\times 2^{\X\cup \Y}\rightarrow S$ is the transition function, and $F$ is the set of accepting states. A game on $\F$ consists of two players, the controller and the environment. $\X$ is the set of \emph{uncontrollable propositions}, which are under the control of the environment, and $\Y$ is the set of \emph{controllable propositions}, which are under the control of the controller. The DFA game problem is to check the existence of \emph{winning strategy} for the controller and generate it if exists. A strategy for the controller is a function $g: (2^{\X})^* \rightarrow 2^{\Y}$, deciding the values of the controllable variables for every possible history of the uncontrollable variables. To define a \emph{winning strategy} 
for the controller, we use the definition of \emph{winning state} below. 

\begin{definition}[Winning State]
	Given a DFA $\F = (2^{\X\cup \Y}, S, s_0, \delta, F)$, $s\in S$ is a \emph{winning state} if $s\in F$ is an accepting state, or there exists $Y\in 2^{\Y}$ such that, for every $X \in 2^\X$, $\delta(s, X\cup Y)$ is a winning state. Such $Y$ is the \emph{winning output} of winning state $s$.
\end{definition}
	
A strategy g is a \emph{winning strategy} if, given an infinite sequence $\lambda = X_0,X_1,X_2,\ldots\in (2^{\X})^\omega$, there is a finite trace $\rho = (X_0\cup g(\epsilon)), (X_1\cup g(X_0)),\ldots , (X_k\cup g(X_0,X_1,\ldots,X_{k-1}))$ that ends at an accepting state. If the initial state $s_0$ is a wining state, then a winning strategy exists. After obtaining a DFA from the \ltlf specification, we can utilize the solution to the DFA game for \ltlf synthesis.

Formally, the \emph{winning strategy} can be represented as a deterministic finite transducer, defined as below.	
\begin{definition}[Deterministic Finite Transducer]\label{extrans}
	Given a DFA $\F=(2^{\X\cup \Y}, S, s_0, \delta, F)$, a deterministic finite transducer $\mathcal{T} = (2^{\X}, 2^{\Y}, Q, s_0, \varrho, \omega, F)$ is defined as follows:
	\begin{itemize}
		\item $Q\subseteq S$ is the set of winning states; 
		\item $\varrho : Q\times 2^{\X}\rightarrow Q$ is the transition function such that $\varrho (q, X) = \delta (q, X\cup Y)~and~Y =\omega(q)$;
		\item $\omega: Q\rightarrow 2^{\Y}$ is the output function such that $\omega (q)$ is a winning output of q.
	\end{itemize}
\end{definition}

\subsection{Boolean Synthesis}
In this paper, we utilize the \emph{boolean synthesis} technique proposed in~\cite{DrLu.cav16}. 
The formal definition of the boolean synthesis problem
is as follows.


\begin{definition}[Boolean Synthesis~\cite{DrLu.cav16}]
Given two disjoint atom sets $\X, \Y$ of input and output variables, respectively, and a boolean formula $\xi$ over $\X\cup \Y$, the boolean synthesis problem is to construct a function $\gamma:2^{\X}\rightarrow 2^{\Y}$ such that, for all $X \in 2^\X$, if there exists $Y \in 2^\Y$ such that $X \cup Y \models \xi$, then $X \cup \gamma(X) \models \xi$. 
The function $\gamma$ is called the \emph{implementation function}.
\end{definition}

We treat boolean synthesis as a black box, using the implementation function construction in the \ltlf synthesis to obtain the output function of the transducer. For more details on techniques and algorithms for boolean synthesis we refer to \cite{DrLu.cav16}. 


\section{Translation from \ltlf Formulas to DFA}\label{sec:dfa}

Following~\cite{DegVa15}, in order to use DFA games to solve the synthesis problem, we need to first convert the \LTLf specification to a DFA. This section focuses on DFA construction. Given an \LTLf formula $\phi$, the corresponding DFA can be constructed explicitly or symbolically.
\subsection{DFA Construction}
SPOT \cite{spot} is the state-of-the-art platform for conversion from \LTL formulas to explicit Deterministic B\"uchi Automaton (DBA). The reduction rules from an \ltlf formula $\phi$ to an \LTL formula $\phi_e$ are proposed in \cite{DegVa13}, and are already implemented in SPOT. Thus by giving an \ltlf formula to SPOT, it returns the DBA $D_e$ for $\phi_e$. $D_e$ can be trimmed to a DFA that recognizes the language of the \LTLf formula $\phi$. For more details on this reduction, we refer to~\cite{DuttaVT13,DuttaV14}.

MONA \cite{KlaEtAl:Mona} is a tool that translates from the Weak Second-order Theory of One or Two successors (WS1S/WS2S) to symbolic DFA. First Order Logic (FOL) on finite words, which is a fragment of WS1S, has the same expressive power as \ltlf, so an \ltlf formula $\phi$ can be translated to a corresponding FOL formula $fol_\phi$ \cite{DegVa13}. Taking such a FOL formula $fol_\phi$ as input, MONA is able to generate the DFA for $\phi$.

\subsection{Evaluations}
It is unnecessary to compare the outputs of SPOT and MONA in terms of size, since both tools return a minimized DFA.The key point is to test them in scalability. 
\LTL and \LTLf have the same syntax, so we construct our benchmarks from 20 basic cases, half of which are realizable, from the \LTL literature \cite{Lily}. Regarding the \LTL benchmarks, the semantics of \LTL formulas is not, in general preserved, when moving to 
the finite-trace setting.

Since these basic cases are too small to be used individually to evaluate the DFA construction tools, we use a class of \emph{random conjunctions} over basic cases~\cite{DanieleGV99}. Note that real specifications typically consist of many temporal properties, whose conjunction ought to be realizable. Formally, a random conjunction formula $RC(L)$ has the form: $RC(L) = \bigwedge_{1\leq i\leq L}P_i(v_1,v_2,...,v_k)$, where $L$ is the number of conjuncts, or the length of the formula, and $P_i$ is a randomly selected basic case (out of the 20 ones). Variables $v_1,v_2,...,v_k$ are chosen randomly from a set of $m$ candidate variables. Given $L$ and $m$ (the size of the candidate variable set), we generate a formula $RC(L)$ in the following way: (1) Randomly select $L$ basic cases; (2) For each case $\phi$, substitute every variable $v$ with a random new variable $v'$ chosen from $m$ atomic propositions.
If $v$ is an X-variable in $\phi$, then $v'$ is also an X-variable in $RC(L)$. The same applies to the Y-variables.

Each candidate variable may be chosen multiple times, so the number of variables in the formula varies. We generate 50 random formulas for each configuration $(L, m)$, adding up to 4500 instances in total. 
Formula lengths $L$ range from 1 to 10, and $m$ varies in increments of 50 from 100 to 500. The platform used in the experiments is a computer cluster consisting of 2304 processor cores in 192 Westmere nodes (12 processor cores per node) at 2.83 GHz with 4GB of RAM per core, and 6 Sandy Bridge nodes of 16 processor cores each, running at 2.2 GHz with 8GB of RAM per core. Time out was set to 120 seconds. Cases that cannot generate the DFA within 120 seconds fail even if the time limit is extended, due to running out of memory.

\begin{figure}
  \centering
  \includegraphics[width=3in]{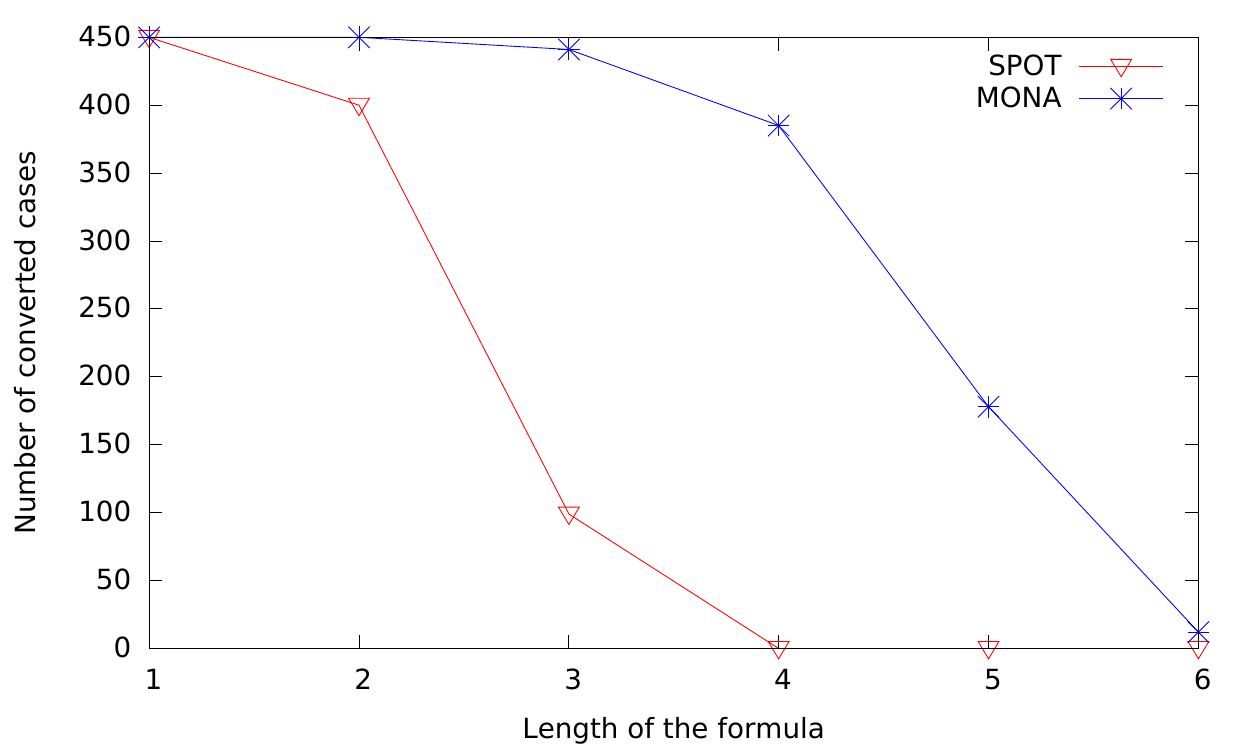}
  \caption{Comparison of scalability between SPOT and MONA on the length of formulas}\label{fig:plotlength}
\end{figure}

Here we consider the number of successfully converted cases for scalability evaluation. The results are summarized in Figure~\ref{fig:plotlength} and ~\ref{fig:plotnvars}, which present the scalability of SPOT and MONA on $L$ and $m$ respectively. As $L$ grows, MONA demonstrates greater scalability, since SPOT cannot handle cases for $L>4$. We conjecture that the poor scalability on $L$ of SPOT is due to the bad handling of conjunctive goals. In the comparison of scalability on $m$, MONA is able to solve around twice as many cases  than SPOT for each $m$. Given these results, we adopt MONA for the DFA construction process and pursue a symbolic approach for \ltlf synthesis.
\begin{figure}
  \centering
  \includegraphics[width=3in]{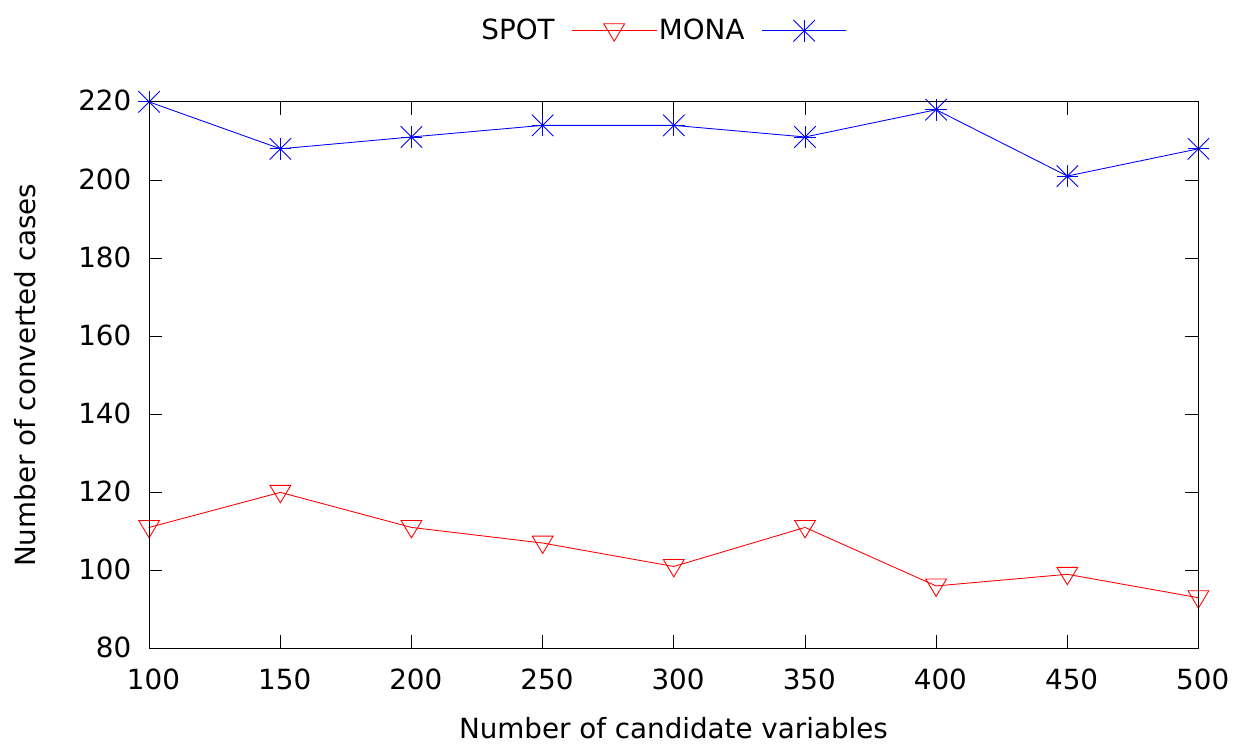}
  \caption{Comparison of scalability between SPOT and MONA on the number of variables}\label{fig:plotnvars}
\end{figure}

\section{Symbolic \ltlf Synthesis}\label{sec:booleansyn}

From an explicit automaton, the DFA game can be solved following the approach described in \cite{DegVa15}, by searching the state graph to compute the set of winning states and choosing for each winning state a winning output. The winning states and outputs can then be used to construct a transducer that implements the winning strategy.

To solve a DFA game symbolically, we first use MONA to construct a symbolic representation of the DFA. Then, we perform a fixpoint computation over this symbolic representation to obtain a boolean formula encoding all winning states of the DFA along with their corresponding winning outputs. Finally, if the game is realizable, we can synthesize a winning strategy from this formula with the help of a boolean synthesis procedure. In the rest of this section we describe each of these steps in more detail.

We start by defining the concept of \emph{symbolic automaton}:
\begin{definition}[Symbolic Automaton] Given a DFA $\mathcal{G} = (2^{\X\cup\Y}, S, s_0, \delta, F)$, the corresponding symbolic automaton $\mathcal{F} = (\mathcal{X}, \mathcal{Y}, \mathcal{Z}, Z_0, \eta, f)$ is defined as follows:
\begin{itemize}
\item $\mathcal{X}$ and $\mathcal{Y}$ are as defined for $\mathcal{G}$;

\item $\mathcal{Z}$ is a set of $\lceil \log_2|S| \rceil$ new propositions such that every state $s \in S$ corresponds to an interpretation $Z \in 2^\mathcal{Z}$;

\item $Z_0 \in 2^\mathcal{Z}$ is an interpretation of the propositions in $\mathcal{Z}$ corresponding to the initial state $s_0$;

\item $\eta : 2^\mathcal{X} \times 2^\mathcal{Y} \times 2^\mathcal{Z} \rightarrow 2^\mathcal{Z}$ is a boolean function mapping interpretations $X$, $Y$ and $Z$ of the propositions of $\mathcal{X}$, $\mathcal{Y}$ and $\mathcal{Z}$ to a new interpretation $Z'$ of the propositions of $\mathcal{Z}$, such that if $Z$ corresponds to a state $s \in S$ then $Z'$ corresponds to the state $\delta(s, X \cup Y)$;

\item $f$ is a boolean formula over the propositions in $\mathcal{Z}$, such that $f$ is satisfied by an interpretation $Z$ iff $Z$ corresponds to a final state $s \in F$.
\end{itemize}

\end{definition}

Intuitively, the \emph{symbolic automaton} represents states by propositional interpretations, the transition function by a boolean function and the set of final states by a boolean formula.

To solve the realizability problem over a symbolic automaton we compute a boolean formula $w$ over $\mathcal{Z}$ that is satisfied exactly by those interpretations that correspond to winning states. The specification is realizable if and only if $Z_0$ satisfies $w$. To solve the synthesis problem, we compute a boolean function $\tau : 2^\mathcal{Z} \rightarrow 2^\mathcal{Y}$ such that for any sequence $(X_0, Y_0, Z_0), (X_1, Y_1, Z_1), \ldots$ that satisfies: (1) $Z_0$ is the initial state; (2) For every $i \geq 0$, $Y_i = \tau(Z_i)$; (3) For every $i \geq 0$, $Z_{i+1} = \eta(X_i, Y_i, Z_i)$ 
there exists an $i$ such that $Z_i$ satisfies $f$. In other words, starting from the initial state, for any sequence of uncontrollable variables, if the controllable variables are computed by $\tau$ and the next state is computed by $\eta$, the play eventually reaches an accepting state.

\subsection{Realizability and Synthesis over Symbolic Automata}

We can compute $w$ and $\tau$ through a fixpoint computation over two boolean formulas: $w_i$, over the set of propositions $\mathcal{Z}$, and $t_i$, over $\mathcal{Z} \cup \mathcal{Y}$. These formulas encode winning states and winning outputs in the following way: every interpretation $Z \in 2^\mathcal{Z}$ such that $Z \models w_i$ corresponds to a winning state, and every interpretation $(Z, Y) \in 2^\mathcal{Z} \times 2^\mathcal{Y}$ such that $(Z, Y) \models t_i$ corresponds to a winning state together with a winning output of that state. When we reach a fixpoint, $w_i$ should encode all winning states and $t_i$ all pairs of winning states and winning outputs.

In the procedure below, we compute the fixpoints of $w_i$ and $t_i$ starting from $w_0$ and $t_0$. We assume that we are able to perform basic Boolean operations over the formulas, as well as substitution, quantification and testing for logical equivalence of two formulas.

In the first step of the computation, we initialize $t_0(Z, Y) = f(Z)$ and $w_0(Z) = f(Z)$, since every accepting state is a winning state. Note that $t_0$ is independent of the propositions from $\mathcal{Y}$, since once the play reaches an accepting state the game is over and we don't care about the outputs anymore.  Then we construct $t_{i+1}$ and $w_{i+1}$ as follows:
$t_{i+1}(Z, Y) = t_i(Z, Y) \lor (\neg w_i(Z) \land \forall X . w_i(\eta(X, Y, Z)))$, $w_{i+1}(Z) = \exists Y . t_{i+1}(Z, Y)$

An interpretation $(Z, Y) \in 2^\mathcal{Z} \times 2^\mathcal{Y}$ satisfies $t_{i+1}$ if either: $(Z, Y)$ satisfies $t_i$; or $Z$ was not yet identified as a winning state, and for every input $X$ we can move from $Z$ to an already-identified winning state by setting the output to $Y$. Note that it is important in the second case that $Z$ has not yet been identified as a winning state, because it guarantees that the next transition will move closer to the accepting states. Otherwise, it would be possible, for example, for $t_{i+1}$ to accept an assignment to $Y$ that moves from $Z$ back to itself, making the play stuck in a self loop.

From $t_{i+1}$, we can construct $w_{i+1}$ by existentially quantifying the output variables. This means that $w_{i+1}$ is satisfied by all interpretations $Z \in \mathcal{Z}$ that satisfy $t_{i+1}$ for some output, ignoring what the output is.
The computation reaches a fixpoint when $w_{i+1} \equiv w_i$ ($\equiv$ denoting logical equivalence). At this point, no more states will be added, and so all winning states have been found. By evaluating $w_i$ on $Z_0$ we can know if there exists a winning strategy. If that is the case, $t_i$ can be used to compute this strategy. This can be done through the mechanism of boolean synthesis.

By giving $t_i$ as the input formula to a boolean synthesis procedure, and setting $\mathcal{Z}$ as the input variables and $\mathcal{Y}$ as the output variables, we get back a function $\tau : 2^\mathcal{Z} \rightarrow 2^\mathcal{Y}$ such that $(Z, \tau(Z)) \models t_i$ if and only if there exists $Y \in 2^\mathcal{Y}$ such that $(Z, Y) \models t_i$.

Using $\tau$, we can define a \emph{symbolic transducer} $\mathcal{H}$ corresponding to the winning strategy of the DFA game. 
\begin{definition}[Symbolic Transducer]\label{symtrans}
Given a symbolic automaton $\mathcal{F} = (\mathcal{X}, \mathcal{Y}, \mathcal{Z}, Z_0, \eta, f)$ and a function $\tau : 2^\mathcal{Z} \rightarrow 2^\mathcal{Y}$, the symbolic transducer $\mathcal{H} = (\mathcal{X}, \mathcal{Y}, \mathcal{Z}, Z_0, \zeta, \tau, f)$ is as follows:
\begin{itemize}
\item $\mathcal{X}$, $\mathcal{Y}$, $\mathcal{Z}$, $Z_0$ and $f$ are as defined for $\mathcal{F}$;
\item $\tau$ is as defined above;
\item $\zeta: 2^{\mathcal{Z}} \times 2^{\mathcal{X}} \rightarrow 2^{\mathcal{Z}}$ is the \emph{transition function} such that $\zeta(Z,X) = \eta(X,Y,Z)$ and $Y = \tau(Z)$.
\end{itemize}
\end{definition}
Note that $\tau$ selects a single winning output for each winning state. Such a transducer is a solution to the DFA game, and therefore to the \LTLf synthesis problem. 

The following theorem states the correctness of the entire synthesis procedure:

\begin{theorem}
If an \ltlf formula $\phi$ is realizable with respect to $\langle\X, \Y\rangle$, the symbolic transducer $\H$ corresponds to a winning strategy for $\phi$.
\end{theorem}

\begin{proof}
The translation from \ltlf to FOL is correct from~\cite{DegVa13}. Correctness of the construction of the symbolic DFA from the FOL formula follows from the correctness of MONA~\cite{KlaEtAl:Mona}. The winning state computation is correct due to being an implementation of the algorithm in~\cite{DegVa15}. The synthesis of $\tau$ is correct from the definition of boolean synthesis.
\end{proof}

\subsection{\LTLf Synthesis via \LTL Synthesis} \label{sec:reduction}
An alternative procedure to \LTLf synthesis can be obtained by a reduction to \LTL synthesis. This reduction allows tools for general \LTL synthesis to be used in solving the \LTLf synthesis problem.
The reduction works as follows. 1) Given an \LTLf formula $\phi$ over propositions $P$, there is an \LTL formula $t(\phi)$ over propositions $P\cup\{Tail\}$ where $Tail$ is a new variable, and $\phi$ is satisfiable iff $t(\phi)$ is satisfiable (see \cite{DegVa13,JZPVH14}). 2) If $P = \X\cup\Y$, and $\X$ and $\Y$ are the set of input and output variables of $\phi$, respectively, set $\X$ and $\Y \cup \{Tail\}$ as the input and output variables of $t(\phi)$, respectively.
Intuitively, the $Tail$ variable is an output variable that indicates when the finite trace should end. As such, when the \LTLf property is satisfied, the controller can set this variable to $true$. The result of solving the \LTL synthesis problem for $t(\phi)$ corresponds to solving the \LTLf synthesis problem for $\phi$. The correctness of the reduction is guaranteed by the following theorem, which follows from the construction in~\cite{DegVa13} together with $Tail$ being an output variable. 

\begin{theorem}
$\phi$ is realizable with respect to $\langle\X, \Y\rangle$ if and only if $t(\phi)$ is realizable with respect to $\langle\X, \Y\cup\{Tail\}\rangle$.
\end{theorem}

In Section~\ref{sec:experiments}, we use this reduction to compare our \LTLf synthesis implementation with tools for standard \LTL synthesis over infinite traces.

\section{Implementation}
We implemented the symbolic synthesis procedure described in Section~\ref{sec:booleansyn} in a tool called \emph{Syft}, using Binary Decision Diagrams (BDDs) to represent the boolean formulas. For comparison, we also implemented the explicit version in a tool called \emph{E-Syft}. Both were implemented in C++11, and \emph{Syft} uses CUDD-3.0.0 as the BDD library.

Each tool consists of two parts: DFA construction and the synthesis procedure.
Based on the evaluations of the performance of DFA construction in Section \ref{sec:dfa}, we adopt MONA to construct the DFA to be given as input to the synthesis procedure. The DFA is given by MONA as a \emph{Shared Multi-terminal BDD} (ShMTBDD)~\cite{Bryant92,BieKlaRau96}, which represents the function $\delta : S \times 2^{\mathcal{X} \cup \mathcal{Y}} \rightarrow S$.  
A ShMTBDD is a binary decision diagram with $|S|$ roots and $m$ terminal nodes ($m \leq |S|$) representing states in the automaton. 
Formally speaking, $\delta (s, X\cup Y) = s'$ is a transition in the DFA if and only if starting from the root representing state $s$ and evaluating the interpretation $X \cup Y$ leads to the terminal representing state $s'$. To evaluate an interpretation $X \cup Y$ on a ShMTBDD, take the high branch in every node labeled by a variable $v \in X \cup Y$ and the low branch otherwise.
We next describe how we preprocess the DFA given by MONA.

\subsection {Preprocessing the DFA of MONA}

\textbf{From ShMTBDD to Explicit DFA}.
Each root in an ShMTBDD corresponds to an explicit state in the DFA.
Moreover, a root of the ShMTBDD includes the information about if the state is initial or accepting, thus enabling $s_0$ and $F$ to be easily extracted for the explicit DFA. To construct the transition function, 
we enumerate all paths in the ShMTBDD, each 
path from $s$ (root) to $t$ (terminal node) corresponding to one transition from state $s$ to $t$ in the DFA. Note that the size of the explicit DFA may be exponential on the size of the ShMTBDD.

\textbf{From ShMTBDD to BDD}. 
Following Section \ref{sec:booleansyn}, we can construct a symbolic automaton in which the transition function $\eta$  is in form of a (multi-rooted) BDD that describes a boolean function $\eta : 2^\mathcal{X} \times 2^\mathcal{Y} \times 2^\mathcal{Z} \rightarrow 2^\mathcal{Z}$. Thus we need to first generate the BDD for the given ShMTBDD.

The basic idea is as follows: (1) From the ShMTBDD of $\delta$, construct a Multi-Terminal BDD (MTBDD) for $\delta'  : 2^\mathcal{Z} \times 2^\mathcal{X} \times 2^\mathcal{Y} \rightarrow S$ with $\lceil \log_2|S|\rceil$ new boolean variables encoding the states, where every path through the state variables, representing an interpretation $Z_s$, leads to the node under the root $s$ in the ShMTBDD. Then, for each transition in $\delta$, there exists an equivalent transition in $\delta'$. (2) Decompose the MTBDD into a sequence of BDDs $\mathcal{B} = \langle B_0, B_1,...,B_{n-1} \rangle$, $n = \lceil \log_2|S|\rceil$, where each $B_i$, when evaluated on an interpretation $(X \cup Y \cup Z)$, computes the $i$-th bit in the binary encoding of state $\delta'(X,Y,Z)$.

The idea of splitting the ShMTBDD into BDDs is illustrated on Figure \ref{fig:shmtbdd2bdd}. As shown in this example, bits $b_0, b_1$ are used to denote the four states $s_0, s_1, s_2, s_3$. In step (1), root $s_0$ is substituted by $Z_{s_0}$ that corresponds to the formula ($\neg b_0 \wedge \neg b_1$). After replacing all roots with corresponding interpretations, the MTBDD is produced. In step (2), $s_0, s_1, s_2, s_3$ can be represented by $00, 01, 10, 11$ respectively, where $b_0$ denotes the leftmost bit. Bit $b_0$ for both $s_0$ and $s_1$ is $0$. So by forcing all paths that proceed to terminals $s_0$ and $s_1$ in the MTBDD to reach terminal node $0$, and all paths to terminals $s_2$ and $s_3$ to reach terminal node $1$, BDD $B_0$ is generated. BDD $B_1$ is constructed in an analogous way for bit $b_1$.
\begin{figure}[t]
  \centering
  \includegraphics[width=2.84in]{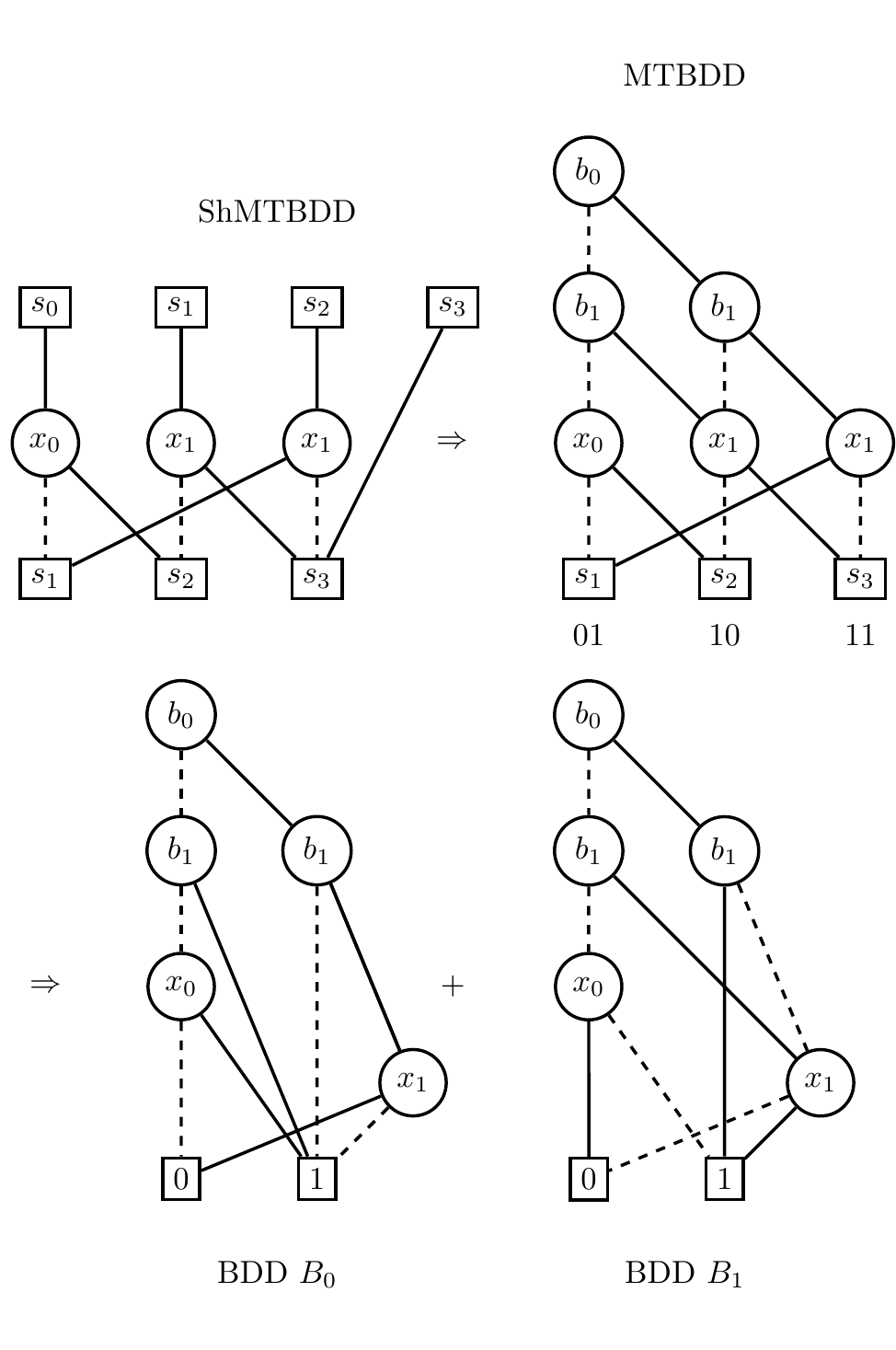}\\
  \caption{Transformation from ShMTBDD to BDD}\label{fig:shmtbdd2bdd}
\end{figure}

\subsection {Implementing the Synthesis Algorithms}

\noindent\textbf{Explicit Synthesis}.
Following \cite{DegVa15}, the main algorithm for explicit synthesis is as follows: starting from accepting states in $F$, iteratively expand the set of winning states. For every state $s$ that is not yet a winning state, if we find an assignment $Y$ such that for all assignments $X$ of input variables, $\delta(s,X \cup Y)$ is a winning state, then $Y$ is set to be the winning output of the new winning state $s$. All assignments of $Y$ have to be enumerated in the worst case. Fixpoint checking is accomplished by checking that no new winning state was added during each iteration. The transducer is generated according to Definition~\ref{extrans}.

\noindent\textbf{Symbolic Synthesis}.
The input here is a symbolic automaton $\mathcal{F}$, in which the transition relation $\eta$ is represented by a sequence $\mathcal{B} =\langle B_0, B_1,...,B_{n-1}\rangle$ of BDDs, where each $B_i$ corresponding to bit $b_i$, and the formula $f$ for the accepting states, is represented by BDD $B_f$. We separate the DFA game into two phases, \emph{realizability} and \emph{strategy construction}.

Following the theoretical framework in Section \ref{sec:booleansyn}, from the accepting states $B_f$, we construct two BDD sequences $\mathbf{T}$ and $\mathbf{W}$, such that $\mathbf{T} = \langle B_{t_0}, B_{t_1},...,B_{t_{i-1}}, B_{t_i}\rangle$ and $\mathbf{W} = \langle B_{w_0}, B_{w_1},...,B_{w_{i-1}}, B_{w_i}\rangle$, where $B_{t_i}$ and $B_{w_i}$ are the BDDs of the boolean formulas $t_i$ and $w_i$ respectively. The fixpoint computation terminates as soon as $B_{w_{i+1}} \equiv B_{w_i}$. 

Our implementation uses the CUDD BDD library~\cite{cudd}, where  fundamental BDD operations are provided. In realizability checking, $B_{t_{i+1}}$ is constructed from $B_{w_i}$ by first substituting each bit $b_i$ of the binary encoding of the state with the corresponding BDD $B_i$. The BDD operation \textit{Compose} is used for such substitution. CUDD also provides the operations \textit{UnivAbstract} and \textit{ExistAbstract} for universal and existential quantifier elimination respectively.
For fixpoint checking, we use canonicity of BDDs, which reduces equivalence checking to constant-time BDD equality comparison. To check realizability, the fixpoint BDD $B_{w_i}$ is evaluated on interpretation $Z_0$ of state variables, returning $1$ if realizable. 

Since veision 3.0.0, CUDD includes a built-in boolean synthesis method \emph{SolveEqn}, which we can use to generate the function $\tau$. From $\tau$, the symbolic transducer can be constructed according to Definition~\ref{symtrans}.

\section{Experiments} \label{sec:experiments}
The experiment was divided into two parts and resulted in two major findings. First, we compared the symbolic approach against the explicit method and showed that the symbolic approach is clearly  superior to the explicit one. Second, we evaluated our \LTLf synthesis implementation against the reduction to \LTL synthesis presented in Section~\ref{sec:reduction}, and again show our approach to be more efficient.

We carried out experiments on the same platform as Section \ref{sec:dfa}.
For the synthesis experiments, besides the original 20 basic cases we also collected 80 instances from the \LTL synthesis tool $Acacia+$~\cite{BohyBFJR12}, making 100 cases in total. In this section,
$L$ denotes the length of the formulas and $N$ denotes the approximate number of variables.
Due to the construction rules of the formula, as described in Section \ref{sec:dfa}, each variable is chosen randomly and may be chosen multiple times. Thus, the exact number of variables in a formula for a given $N$ varies in the range $(N-5, N+5]$.
\subsection{Symbolic vs Explicit}
We aim to compare the results on two aspects: 1) the scalability of each approach; 2) the performance on each procedure (automata construction/safety game) of \LTLf synthesis.

\textbf{\emph{Scalability on the length of formulas.}}
We evaluated the scalability of \emph{Syft} and \emph{E-Syft} on 2000 benchmarks where the formula length $L$ ranges from 1 to 10 and $N=20$, as MONA is more likely to succeed in constructing the DFA for this value of $N$. Each $L$ accounts for 200 of the 2000 benchmarks. Figure~\ref{fig:synlength} shows the number of cases the tools were able to solve for each $L$. As can be seen, the symbolic method can handle a larger number of cases than the explicit one, which demonstrates that the symbolic method outperforms the explicit approach in scalability on the length of the formula.
\begin{figure}
  \centering
  \includegraphics[width=3in]{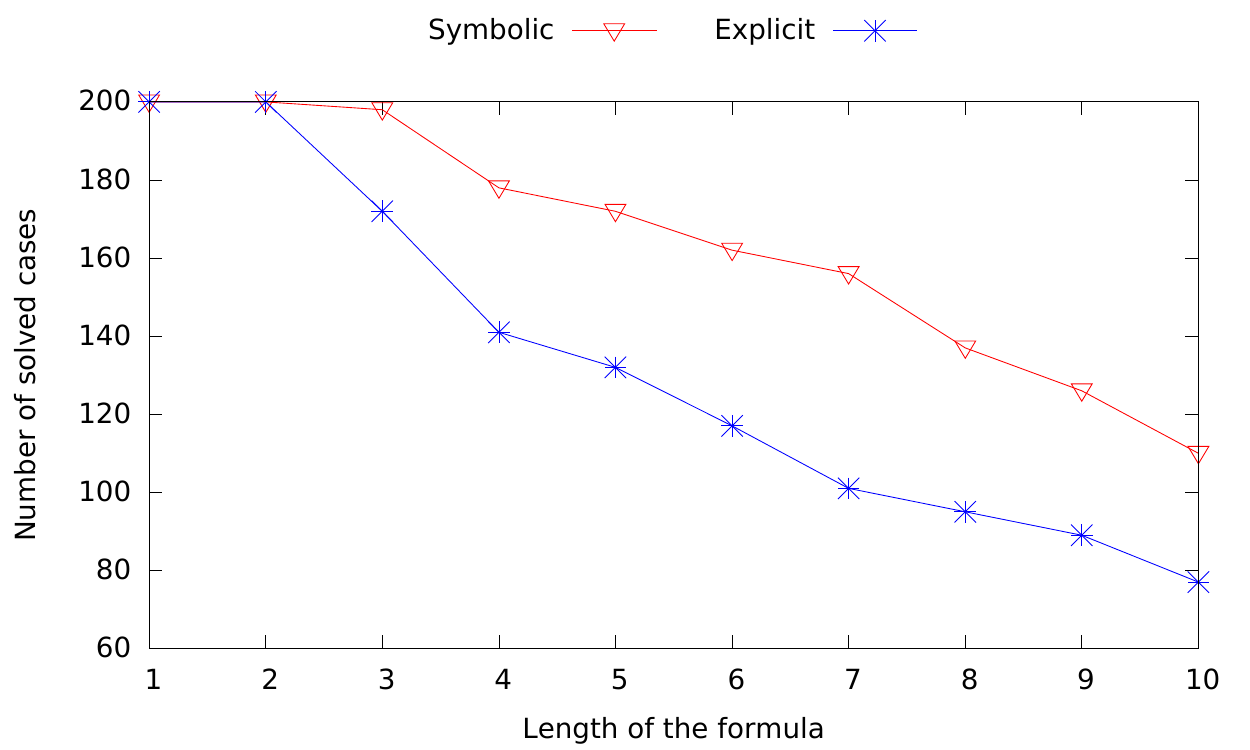}
  \caption{Comparison of scalability of Symbolic against Explicit on the length of formula}\label{fig:synlength}
\end{figure}

\textbf{\emph{Scalability on the number of variables.}}
In this experiment, we compared the scalability over the number of variables $N$, which varies from 10 to 60, where $L$ is fixed as 5. As can be seen in Figure \ref{fig:synvars}, for smaller cases the two approaches behave similarly, since they succeed in almost all cases. The same happens for larger cases, because MONA cannot construct the DFAs. For intermediate values the difference is more noticeable, showing better performance by the symbolic tool. The number of cases that the explicit method can solve sharply declines when $N=30$. However, the symbolic tool can handle more than 40 variables. Both methods tend to fail for $N>50$ at the DFA conversion stage.
\begin{figure}[t]
  \centering
  \includegraphics[width=3in]{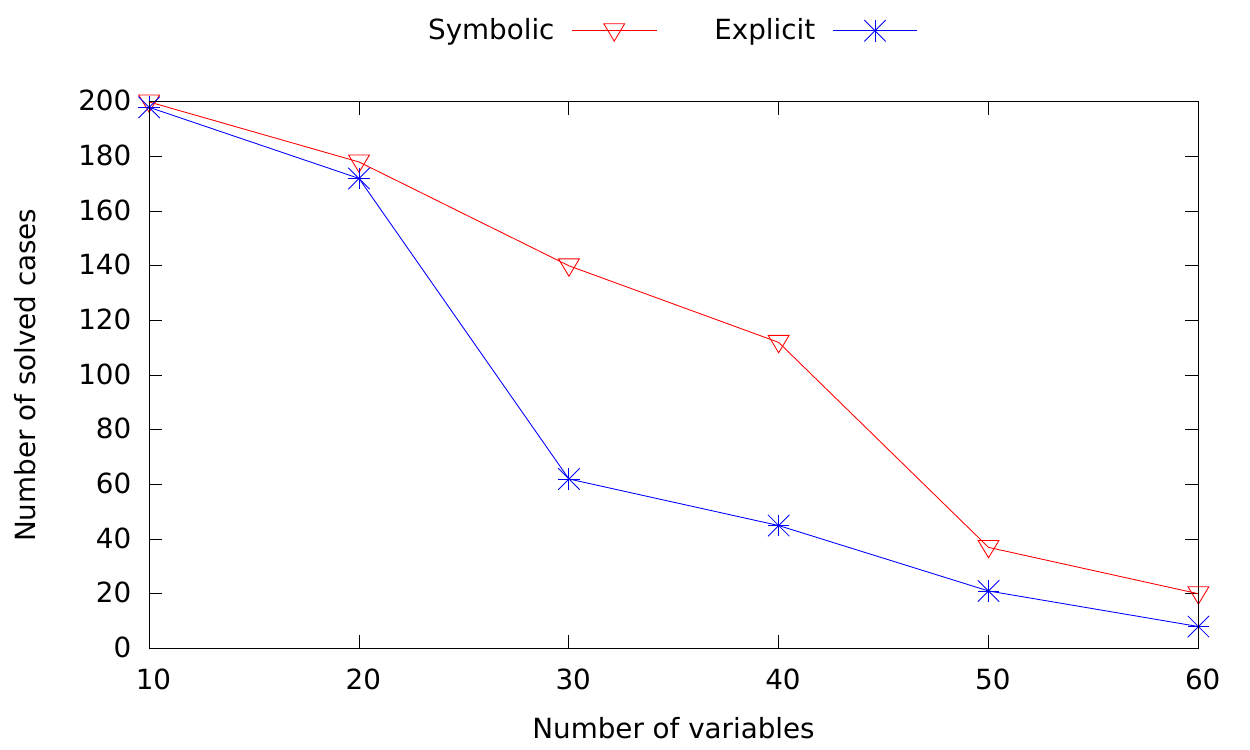}
  \caption{Comparison of scalability of Symbolic against Explicit on the number of variables}\label{fig:synvars}
\end{figure}

\textbf{\emph{Synthesis~vs~DFA Construction.}} 
In this experiment we studied the effect of $N$ on the time consumed by DFA construction and synthesis.
The percentage of time consumed by each is shown in Figure \ref{fig:dfasyn}. We observe that for large cases, DFA construction dominates the running time. As the size of the DFA increases, DFA construction takes significantly longer, while synthesis time increases more slowly, widening the gap between the two. This result allows us to conclude that the critical performance limitation of synthesis is the DFA construction process rather than synthesis itself. 
\begin{figure}
  \centering
  \includegraphics[width=3in]{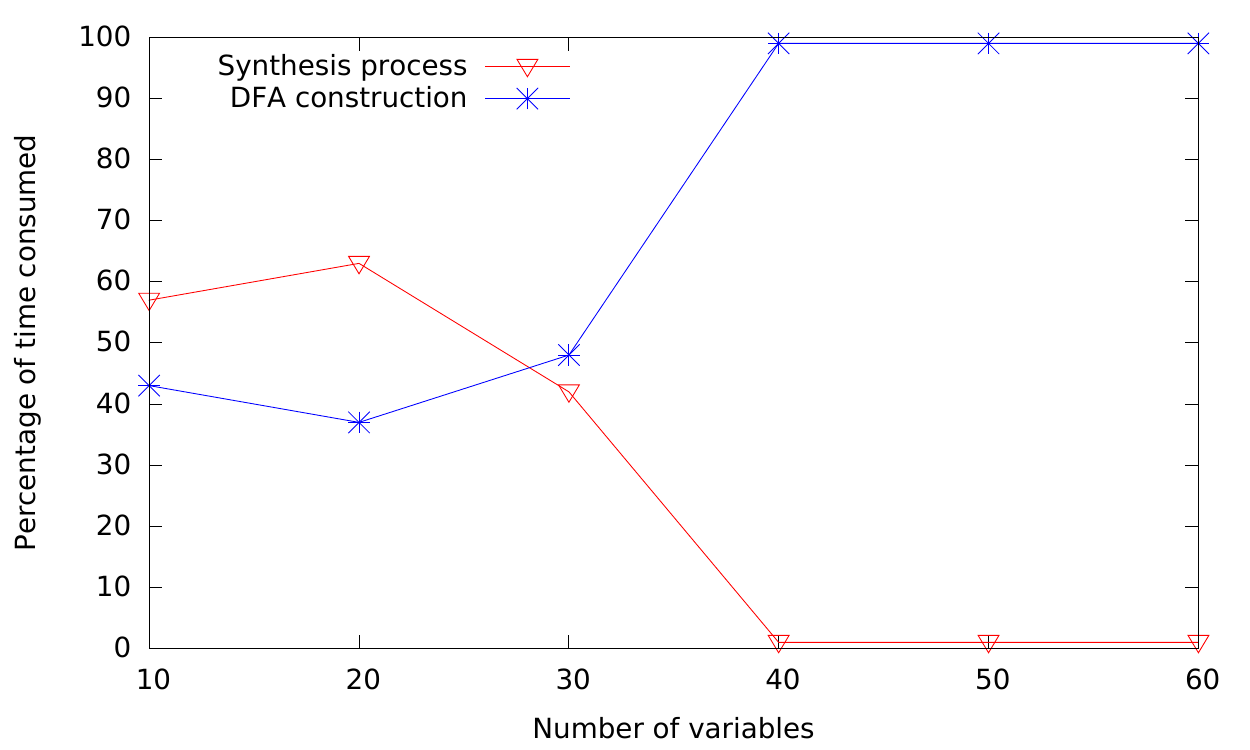}
  \caption{Comparison of scalability of Symbolic against Explicit on the number of variables}\label{fig:dfasyn}
\end{figure}

\subsection{Comparison with Standard \LTL Synthesis}
Here we adopted Acacia+~\cite{BohyBFJR12}, in which LTL2BA~\cite{ltl2ba} is the automaton constructor, as the \LTL synthesis tool and SPOT~\cite{spot} as the translator to convert from \ltlf formulas to \LTL formulas, as presented in  Section~\ref{sec:reduction}. 
We evaluated the effectiveness of each approach in terms of the number of solved cases. Figure~\ref{fig:acacia} shows the number of solved cases as the length of the formula grows. As shown in the figure, Acacia+ is only able to solve a small fraction of the instances that \emph{Syft} can solve.
The evidence here confirms that restricting the problem to finite traces allows more efficient techniques to be used for synthesis.

\begin{figure}[t]
  \centering
  \includegraphics[width=3in]{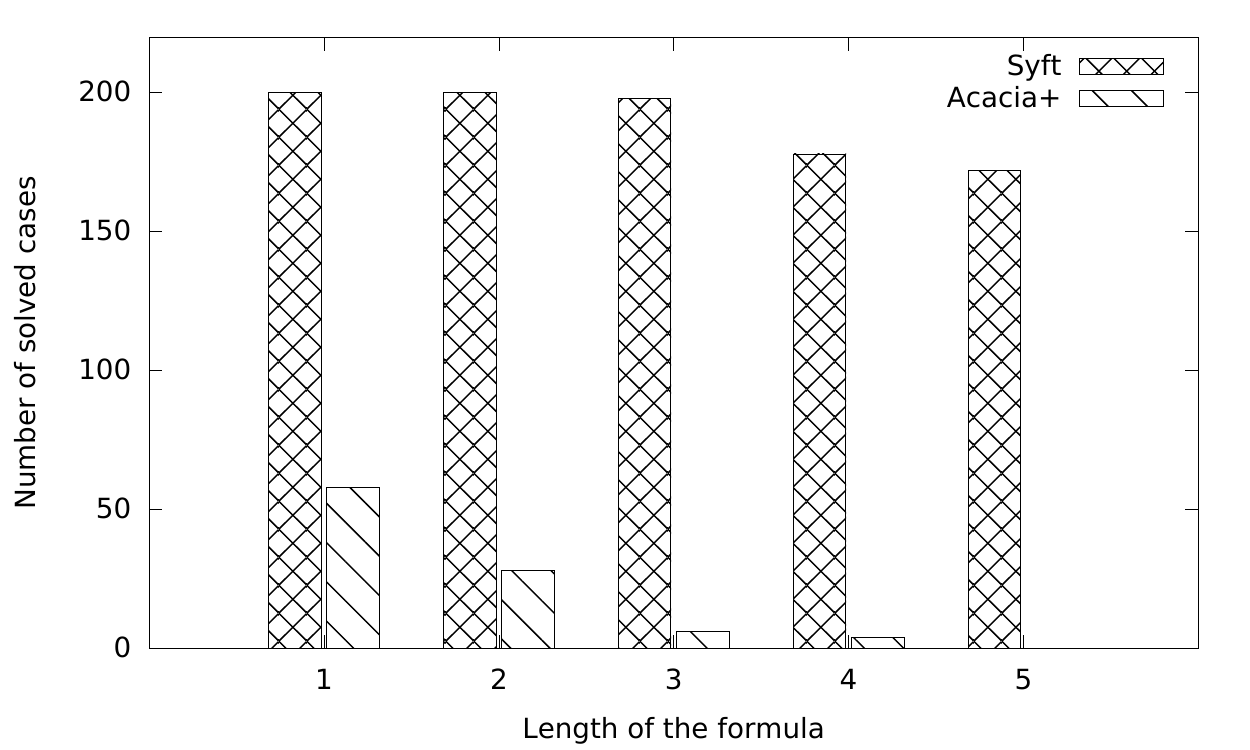}
  \caption{Comparison of scalability of Symbolic approach against Acacia+ on the length of the formula}\label{fig:acacia}
\end{figure}

\subsection{Discussion}
The symbolic synthesis method scales better than the explicit approach both on the length of \ltlf formulas and the number of variables. The performance of the symbolic method, however, relies critically on the DFA-construction process, making this the bottleneck of \ltlf synthesis. Comparing with the reduction from \ltlf synthesis to \LTL synthesis, the advantage of our approach is that the DFA construction can leverage techniques developed for finite-state automata, which cannot be applied to the construction of automata over infinite words, a key step in \LTL synthesis.

\section{Conclusion}

We presented here the first realization of a symbolic approach for \LTLf synthesis, based on the theoretical framework of DFA games. Our experimental evaluation shows that these techniques are far more efficient than a reduction to standard \LTL synthesis. Furthermore, our experiments on DFA construction and synthesis have shown that a symbolic approach to this problem has advantages over an explicit one. In both cases, however, the limiting factor for scalability was DFA construction. When the DFA could be constructed, the symbolic procedure was able to synthesize almost all cases, but for larger numbers of variables DFA construction is not able to scale.

These observations suggest the need for more scalable methods of symbolic DFA construction. A promising direction would also be to develop techniques for performing synthesis ``on the fly'' during the construction of the DFA, rather than waiting for the entire automaton to be constructed before initiating the synthesis procedure.

In~\cite{DegVa13} an alternative formalism for finite-horizon temporal specifications is presented in the form of Linear Dynamic Logic over finite traces, or \LDLf. \LDLf is strictly more expressive than \LTLf, but is also expressible as a DFA. Therefore, given a procedure to perform the conversion from \LDLf to DFA, our approach can be used in the same way. This would allow synthesis to be performed over a larger class of specifications.

\section{Acknowledgment}
Work supported in part by NSFC Projects No.~61572197 and No.~61632005, MOST NKTSP Project~2015BAG19B02, STCSM Project~No.16DZ1100600, project Shanghai Collaborative Innovation Center of Trustworthy Software for Internet of Things~(ZF1213), NSF grants~CCF-1319459 and~IIS-1527668, NSF Expeditions in Computing project~``ExCAPE: Expeditions in Computer Augmented Program Engineering'' and by the Brazilian agency CNPq through the Ci\^{e}ncia Sem Fronteiras program.

\newpage
\bibliographystyle{named}
\bibliography{ijcai17}

\end{document}